\documentclass[11pt,epsf]{article}
\usepackage{epsfig}
\usepackage{latexsym}
\usepackage{amsmath}
\usepackage{booktabs}

\usepackage[all,color,frame,import]{xy}

\newtheorem{theorem}{Theorem}[section]
\newtheorem{lemma}{Lemma}[section]

\newcommand{\tab}{\hspace{.25in}}
\newcommand\QED{\ifhmode\allowbreak\else\nobreak\fi
\quad\nobreak$\Box$\medbreak}
\newcommand{\proofstart}{\par\noindent \emph{Proof:} }
\newcommand{\proofend}{\QED\par}
\newenvironment{proof}{\proofstart}{\proofend}

\def\ignore#1{\relax}

\long\gdef\boxit#1{\vspace{5mm}\begingroup\vbox{\hrule\hbox{\vrule\kern3pt
\vbox{\kern3pt#1\kern3pt}\kern3pt\vrule}\hrule}\endgroup}

\usepackage[letterpaper,hmargin=1in,vmargin=1in]{geometry}

\begin{document}

\title{Approximation Algorithms for the Traveling Repairman and Speeding Deliveryman Problems}
\author{Greg N. Frederickson and Barry Wittman}

\maketitle    

\begin{abstract}
Constant-factor, polynomial-time approximation algorithms are presented
for two variations of the traveling salesman problem with time windows.  
In the first variation, the traveling repairman problem, 
the goal is to find a tour that visits the maximum possible 
number of locations during their time windows.  
In the second variation, the speeding deliveryman problem, 
the goal is to find a tour that uses the minimum possible 
speedup to visit all locations during their time windows.  
For both variations, the time windows are of unit length,
and the distance metric is based on a weighted, undirected graph.  
Algorithms with improved approximation ratios are given for the case
when the input is defined on a tree rather than a general graph.
The algorithms are also extended to handle time windows whose lengths fall in any bounded range.
\end{abstract}

\section{Introduction}

The traveling salesman problem (TSP) has served as the
archetypal hard combinatorial optimization problem that attempts
to satisfy requests spread over a metric space \cite{Lawler}.
Yet, the TSP is not a perfect model of real life.
In particular, a salesman may not have enough time
to visit all desired locations.
Furthermore, a visit to any particular location may be of value
only if it occurs within a certain specified interval of time.
We use the term \emph{repairman problem} to describe the class of problems that add time constraints to the TSP. 

We consider a fundamental version of such a repairman problem,
in which the repairman is presented with a set of
\emph{service requests}.
Each service request is located at a node in a weighted, undirected graph 
and is assigned a \emph{time window} during which it is valid.  Note that multiple service requests may share the same node as a location but have different time windows.
The repairman may start at any time from any location and stop similarly.
(This latter assumption is at variance with
much of the preceding literature about the repairman problem \cite{Bansal,Bar-Yehuda}. 
We choose to frame our problem without specifying initial and final locations 
because doing so leads to an elegant solution that gives additional insight into such problems.)

We handle two variations of our problem.
In the first, when a repairman visits the location of a service request during its time window he performs a \emph{service event}, and each such event
yields a specified \emph{profit}.
A \emph{service run} is a feasible sequence of service events that a repairman can make at a given speed.
The goal of the repairman is to find a service run that satisfies a subset of requests with the maximum total profit possible.
On the other hand, a \emph{service tour} is a service run that satisfies \textbf{all} service requests.
Thus, in this second variation, the service provider tries to minimize the \emph{speed} necessary to make a service tour.
Note that there is some minimum speed below which it is not possible to visit all requests.
We call this variation the \emph{speeding deliveryman problem},
recognizing that, for example, a pizza delivery driver may need to hurry
to deliver his or her set of orders in a timely manner.
We seem to be the first to frame this second problem in terms of speedup, 
a refreshing change from the standard emphasis on distance traveled
or profit achieved.

For both variations,
we focus primarily on the case in which all time windows
are the same length (i.e., unit-time),
and all profits for service events are identical.
Additionally, we refer to each service event as being instantaneous,
although positive service times can be absorbed into the structure of the graph in many cases.
These restrictions still leave problems that are APX-hard for a metric graph,
via a simple reduction from TSP, which has been shown to be APX-hard \cite{Papadimitriou2}.

Our goal is thus to find polynomial-time approximation algorithms.
For the repairman, our algorithms produce a service run whose profit is within
a constant factor of the profit for an optimal service run.
For the deliveryman, they produce a service tour whose maximum speed
is within a constant factor of the optimum speed, which accommodates all requests.
These variations contrast neatly,
as the repairman is a maximization problem
while the deliveryman is a minimization problem.  
To the best of our knowledge, we are the first to find
approximation algorithms for either problem that get within a constant factor for a general metric, albeit when time windows are the same length.  Thus, we establish membership in APX for these specific problem versions.

Our repairman and deliveryman problems are NP-hard even in the case that the service network 
is an edge-weighted tree rather than a general (weighted) graph, as we shall show.
This property is particularly notable, since of course the TSP is polynomial-time solvable on tree networks.
In this simpler context of a tree, we give approximation algorithms with improved constants and faster polynomial running times for our problems.

Although we seem to be the first to study the
speeding deliveryman problem,
we are not the first to consider the repairman problem,
which is a generalization of a host of repairman,
deliveryman and traveling salesman problems such as those in \cite{Bansal,Bar-Yehuda,Chekuri,Karuno3,Tsitsiklis}.  
Much work has been done on related problems in a metric space on the line.
Assuming unit-time windows, a ${4+\epsilon}$-approximation was given
for the repairman on a line in \cite{Bar-Yehuda}.
We improve this approximation to $3$ and in a more general setting, a tree.
We are the first to give poly-time constant-ratio
algorithms for the unit time window repairman problem
on a tree or on a graph.

For general metric spaces and general time windows together in the rooted problem, 
an $O(\log^2 n)$-approximation is given in \cite{Bansal}.  An $O(\log L)$-approximation is given in \cite{Chekuri3}, for the case that all time window start and end times are integers, where $L$ is the length of the longest time window.  In contrast, a constant approximation is given in \cite{Chekuri},
but only when there are a constant number of different time windows.  Following the initial publication of our work in \cite{Frederickson3}, an extension was given in \cite{Chekuri3} that gives an $O(\log D)$-approximation to the unrooted problem with general time windows, where $D$ is the ratio of the length of largest time window to the length of the smallest.  Polylogarithmic approximation algorithms to directed TSP with time windows have been given in \cite{Chekuri2} and \cite{Nagarajan2}.
TSP with time windows has also been studied in the operations research community,
as in \cite{Focacci1} and \cite{Focacci2},
where it is exhaustively solved to optimality.

The problem of orienteering is also significant because it is used as a subroutine in many deadline and time window problems.  In orienteering, the goal is to find a path visiting as many locations as possible, subject to some constraint on the total distance traveled (or time taken).  The first significant results in this area found constant approximations for several variations in the plane \cite{Arkin}.  A PTAS for orienteering in the plane was later given in \cite{Chen}.  Recent results in rooted and point-to-point versions of orienteering \cite{Bansal, Blum3, Chekuri2} have made the latest improvements in approximation algorithms for time window problems possible.

In this paper, we introduce a novel time-partitioning scheme that is especially well suited to unit-length time windows.
We partition requests
into subsets in such a way that we can play off
proximity of location against proximity of time
when forming the subsets.
Partitioning is also done in \cite{Bansal}, but that partitioning approach is different because it is designed to handle general time windows in a fashion not intended to get within a constant factor of optimal, even when the windows are unit-time.
Our approach partitions requests by their time windows,
so that the requests of any subset in the partition
are uniformly available over the entire extent of time
under consideration for that subset.
In our partitioning, we identify discrete \emph{periods} of equal length 
and \emph{trim} the time window for each request
to be the period that was wholly contained in it.  
Trimming induces at most a linear number of periods, each of which we can then consider separately.
Trimming loses the repairman at most a constant fraction of possible profit
and increases the necessary speed of the deliveryman by at most a constant factor.

For the variations restricted to trees, once we partition requests on the basis of common periods,
we are able, for requests with a common period,
to solve a variety of subproblems exactly for the repairman
and almost exactly for the deliveryman, in contrast to general graphs for which we use approximate rather than exact solutions.
For the repairman on a graph, we use constant approximation algorithms from \cite{Bansal} and \cite{Chekuri2} as subroutines.
For all of the problems we consider, we can combine solutions for each different period using dynamic programming.  Although dynamic programming is used for the deliveryman, the algorithm and especially the analysis differ from the repairman.  A key insight is that the effects of trimming can be offset by increasing speed and that the amount of speed needed can be analyzed by imagining the deliveryman running a backwards and forwards pattern along an optimal service tour.

To deal with windows with lengths between 1 and 2 (or between 1 and some constant $c$), we can generalize our repairman algorithms by using more than one trimming scheme.  Each trimming scheme employs a different period size that, when all such schemes are considered together, adapts to different distributions of window size.  By starting each trimming scheme at a number of carefully chosen times and keeping the most profitable run found, we show that the approximation factor for repairman on windows of different lengths can be bounded by a weighted average of the bounds of each trimming scheme.  For windows with length between 1 and 2, this bound yields a constant-factor approximation with a better bound than the result in \cite{Chekuri3} for the same problem.  A different accounting of trimming shows that the speeding deliveryman on windows with length between 1 and 2 can also be approximated to within a constant factor.  Other work \cite{Bansal} has focused on time windows with arbitrary lengths, but improved approximation guarantees for time windows with lengths in some bounded range may be useful for many practical applications in which time window lengths do not vary dramatically.

In Sect.\ \ref{section:trimming}, 
we characterize the effects of contracting the time windows of the service requests. 
In Sect.\ \ref{section:repairman-tree}, we give an approximation algorithm with a bound of $3$
for the repairman on a tree.
In Sect.\ \ref{section:repairman-graph}, we give an approximation algorithm with a bound of $6 + \epsilon$
for the repairman on a graph.
In Sect.\ \ref{section:deliveryman-tree}, we give an approximation algorithm for a deliveryman
on a tree with a maximum increase in speed by a factor of $4+\epsilon$.  We note that the standard notation for approximation ratios may cause confusion in this context where both maximization and minimization problems are being considered, because these ratios are always given as values greater than 1.
In Sect.\ \ref{section:deliveryman-graph}, we give an approximation algorithm for a deliveryman
on a graph with a maximum increase in speed of a factor of 8.
In Sect.\ \ref{section:np-hardness}, we sketch the NP-hardness of the problems on a tree.
In Sect.\ \ref{section:service-times}, we show ways in which non-zero service times for the repairman problem can easily be accommodated with small changes to our algorithms.
In Sect.\ \ref{section:different-repairman}, we extend our repairman algorithms to time windows whose lengths are all within a factor of two of each other and then show how this idea can be applied to time windows with lengths in any bounded range.  
In Sect.\ \ref{section:different-deliveryman}, we extend the analysis for our deliveryman algorithms to time windows whose lengths are all within a factor of two of each other.  Once again, this idea can be expanded to time windows with lengths in any bounded range.

A preliminary version of this paper appeared in \cite{Frederickson3}.

\section{Trimming Requests}
\label{section:trimming}

Trimming is a simple and yet powerful technique
that can be applied when we deal with unit-time windows.
Starting with time 0, we make divisions in time at values 
which are integer multiples of one half, i.e., 0, .5, 1, and so on.  
We assume that no request window starts on such a division,
because if it did, we could redefine times to be decreased by a negligible amount.
We thus assume that the starting time for any window is positive.
Let a \emph{period} 
be the time interval from one division up to but not including the next division.
Because every service request is
exactly one unit long in time, half of any request window will be wholly contained within only one period, with the rest
divided between the preceding and following periods. 
We then trim each service request window to coincide with the period wholly contained in it,
discarding those portions of the request window that fall outside of the chosen period.

For the repairman problem,
the trimming may well lower the profit of the best service run, but by no more than a constant factor.
Let the \emph{target interval} of a request be that part of the request 
window that coincides with the period to which the request is trimmed.  
Call that part of the request window contained in the previous period
its \emph{late interval}, 
and call that part of the request window contained in the following period
its \emph{early interval}.
Let $\pi(R)$ denote the profit of a service run $R$.

\begin{theorem}\label{theorem:trimming}
{\bf (Limited Loss Theorem)}
Consider any instance of the repairman problem.
Let $R^*$ be an optimal service run with respect to untrimmed requests.  
There exists a service 
run $R$ with respect to trimmed requests such that $\pi(R) \geq \frac{1}{3} \ \pi(R^*)$.  
\end{theorem}

\begin{proof}
We use an elegant best-of-three argument.  Observe that $R^*$ must have at least one third of its service events
in either the target intervals, the early intervals, or the late intervals.
If at least one third of the service events of $R^*$ occur in target intervals, 
then have $R$ follow the same path and schedule as $R^*$ but service only those
requests in target intervals.

If at least one third of the service events of $R^*$ occur in late intervals,
then take service run $R$ to be $R^*$ but started .5 units later in time,
and with $R$ servicing those requests that were in late intervals of $R^*$
but are now in target intervals of $R$.  
Then the number of service events of $R$ will be at least one third
of the number of service events for $R^*$.

Similarly, 
if at least one third of the service events of $R^*$ occur in early intervals,
take $R$ to be $R^*$ but started .5 units earlier in time,
with $R$ servicing those requests that were in early intervals of $R^*$
but are now in target intervals of $R$.  Recall that starting $R$ earlier in time is permissible in the unrooted problem.

In each case, 
there is a service run $R$ for trimmed requests 
that contains at least one third of the service events
of an optimal service run for untrimmed requests. 
Since one of these three cases must always hold, 
the desired $R$ always exists.
\end{proof}

For the deliveryman problem,
trimming may well increase the necessary speed of the best service tour, but by no more than a constant factor.  
Let $s(Q)$ denote the minimum speed needed for service tour $Q$ to visit all service requests.

\begin{theorem}
\label{theorem:deliverytrimming}
{\bf (Small Speedup Theorem)}
Consider any instance of the deliveryman problem.
Let $Q^*$ be an optimal service tour with respect to untrimmed requests starting at time $t = 0$.  
There exists a service tour $Q$ with respect to trimmed requests such that $s(Q) \leq 4 s(Q^*)$.  
\end{theorem}

\begin{proof}
We shall extend $Q^*$ backward for $t < 0$ by assuming that $Q^*$ proceeds from any convenient position so that it encounters the original starting position at time $t = 0$.  
Let \emph{racing} describe movement, either forward or backward, along $Q^*$ at a speed of $4 s(Q^*)$.
We define tour $Q$ which races along $Q^*$.  
During any two consecutive periods, the deliveryman 
will make a net advance equal to the advance of $Q^*$ over those two periods. 

Identify as $t_i$ the time $t = .5i$ which is also the starting time of period $i$. Let $f(t)$ be a function that gives the location of the deliveryman on $Q^*$ for any given time $t$.  We define $Q$ as follows.  Start tour $Q$ at $t = 0$ at the location that $Q^*$ has at time $t = -.5$.  From there, tour $Q$ follows a repeating pattern of racing forward along $Q^*$ for $1$ period, racing backward along $Q^*$ for $.75$ periods, and then racing forward along $Q^*$ for $.25$ periods.  We define $q(t)$ to describe the movement of $Q$ as follows.

\noindent For $t_i \leq t < t_i + 1$, where $i$ is even, define\\

\noindent $q(t) = \left\{ \begin{array}{lrclr}
f(t_i - .5 + 4(t - t_i))   & t_i &\leq~t~\leq & t_i + .5 & \mbox{(forward for 1 period)}\\
f(t_i + 3.5 - 4(t - t_i)) & t_i + .5 &\leq~t~\leq &t_i + .875 & \mbox{(backward for .75 periods)}\\
f(t_i -3.5 + 4(t - t_i))  & t_i + .875 &\leq~t~\leq &t_i + 1 & \mbox{(forward for .25 periods)}
\end{array} \right.$\\

\noindent Figure \ref{figure:s=4} gives an example of this pattern of movement for some $Q^*$ and a corresponding $Q$.

\begin{figure}[!hbt]
\centering
\begin{xy}
\xyimport(456, 148){\includegraphics[width=.75\textwidth]{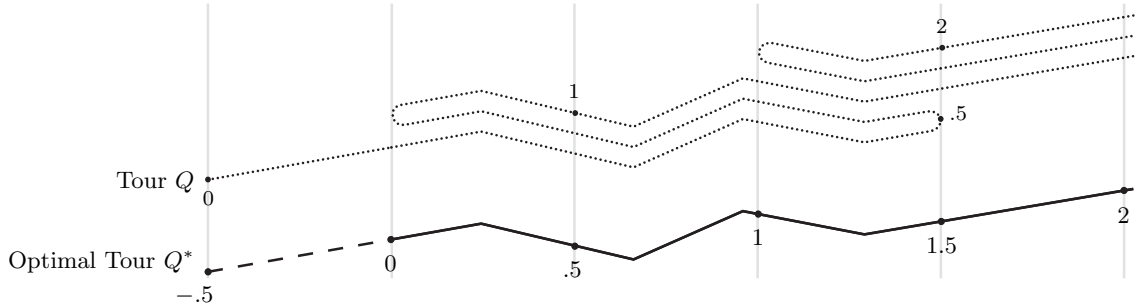}}
(-5,52)*!R\txt\footnotesize{Tour $Q$};
(-5,9)*!R\txt\footnotesize{Optimal Tour $Q^*$};
(-5,-9)*\txt\footnotesize{$-.5$};
(91,8)*\txt\footnotesize{$0$};
(180,5)*\txt\footnotesize{$.5$};
(271,22)*\txt\footnotesize{$1$};
(361,18)*\txt\footnotesize{$1.5$};
(450,34)*\txt\footnotesize{$2$};
(2,43)*\txt\scriptsize{$0$};
(181,101)*\txt\scriptsize{$1$};
(369,89)*\txt\scriptsize{$.5$};
(361,136)*\txt\scriptsize{$2$};
\end{xy}
\caption{Example of tour $Q$ at speedup of 4 compared with an optimal tour $Q^*$.}
\label{figure:s=4}
\end{figure}

Consider a request $r$ serviced at time $t$ in $Q^*$.  If $t_i \leq t < t_i + .5$, 
then the time window of the request will be trimmed to be one of three periods of length .5:
$[t_i - .5, t_i)$, $[t_i, t_i + .5)$, or $[t_i + .5, t_i + 1)$.  We consider cases when $i$ is odd or $i$ is even separately.

\begin{description}
\item Case 1: $i$ is odd

If the window containing $r$ is trimmed to be $[t_i - .5, t_i)$, then service the request $r$ at time $t_i + .25((t - t_i) - 1)$.  
If the window is trimmed to be $[t_i, t_i + .5)$, then service the request $r$ at time $t_i + .25((t_i - t) + 1)$.   
If the window is trimmed to be $[t_i + .5, t_i + 1)$, then service the request $r$ at time $t_i + .25((t - t_i) + 2)$.

\item Case 2: $i$ is even

If the window containing $r$ is trimmed to be $[t_i - .5, t_i)$, then service the request $r$ at time $t_i + .25((t_i - t) - 1.5)$.  
If the window is trimmed to be $[t_i, t_i + .5)$, then service the request $r$ at time $t_i + .25((t - t_i) + .5)$.   
If the window is trimmed to be $[t_i + .5, t_i + 1)$, then service the request $r$ at time $t_i + .25((t_i - t) + 3.5)$.
\end{description}
\end{proof}

\section{Repairman Problem for a Tree}
\label{section:repairman-tree}
Trimming is indeed a valuable technique because 
we can solve the repairman problem on a tree exactly 
in the case when windows are already trimmed.
We first give a dynamic programming algorithm for the repairman problem on a tree 
when all requests share the same time window.
To find a path from $s$ to $t$ of profit $p$,
we start with the direct path from $s$ to $t$ 
and then add on low-cost pieces of subtrees that branch off the direct path
as necessary to achieve profit $p$.
We do so by contracting the path into a single node $r$  
and using dynamic programming to sweep up from the leaves, 
finding the cheapest paths in the tree for each possible profit. 

\begin{table}[!hbt]
\begin{tabular}{l}
\smallskip \\
\toprule
\textbf{SWEEP-TREE(node \emph{u})} \\
\midrule
\tab For $p$ from 0 to $\pi(u)$, set $L_u[p]$ to be $0$.\\
\tab For each child $v$ of $u$,\\
\tab \tab Call SWEEP-TREE($v$), which will generate $L_v$.\\
\tab \tab Add $2d(u, v)$ to each entry in $L_v$ except $L_v[0]$.\\
\tab \tab Let $\max _u$ be the largest profit in $L_u$ and $\max _v$ the largest profit in $L_v$.\\
\tab \tab For $p$ from 0 to $\max _u +\max _v$, set $L[p]$ to be $\infty$.\\
\tab \tab For $a$ from 0 to $\max_u$ and $b$ from $0$ to $\max_v$,\\
\tab \tab \tab Set $L[a+b]$ to be $\min \{ L[a+b], L_u[a]+L_v[b]\}$.\\
\tab \tab Set $L_u$ to be $L$.\\
\bottomrule
\end{tabular}
\end{table}

Our recursive subroutine SWEEP-TREE$( r )$ produces a list $L_r$ of the lowest costs 
at which various profit levels can be achieved by including portions of the tree rooted at $r$.  
List $L_r$ is a mapping from profits to costs where $L_r[p]$ is the cost of achieving profit $p$, if recorded, and $\infty$ otherwise.
Let $\pi(u)$ be the profit gained by visiting $u$. 
Note that $\pi(u)$ counts the number of service requests at $u$. 
If we define $\pi(r)$ to be the profit of the direct path, 
then adding $d(s,t)$ to all the costs in the list $L_r$ yields 
the costs of the best paths on the full tree starting at $s$ and ending at $t$ for all possible profit levels.

\begin{lemma}
\label{lemma:treerepairsame}
For all possible profits, SWEEP-TREE identifies minimum-length paths from $s$ to $t$ in a total of $O(n^2)$ time.
\end{lemma}

\begin{proof}Correctness follows from induction on the size of the tree.  Let $n$ represent the sum of the number
of nodes in the subtree rooted at node $u$ plus the total profit for the nodes in that subtree.
Let $n_0$ be the portion of $n$ that is attributable to node $u$,
and let $n_i$ be the portion of $n$ that is attributable to the subtree rooted at the $i^\mathrm{th}$ child of $u$, for $i=1, 2, \ldots , k$.
Then the time for SWEEP-TREE is described by
$$T(n) \leq cn_0 + \sum_{i=1}^k~\left[~T(n_i) +cn_i\left( 1+ \sum_{j=0}^{i-1}n_j \right)~\right]$$
The first term in this inequality accounts for the time it takes to initialize list $L_u$.  The outer summation accounts for the time spent on each child of $u$: first the recursion on the child subproblem, then the time to update each entry in $L_v$, and finally, with the inner summation, the time to update list $L$.
An induction proof establishes that $T(n) \leq dn^2$ for a suitable constant $d$.
\end{proof}

Although SWEEP-TREE might be viewed as being reminiscent of Sect.\ 2.6.3 in \cite{Chawla}, we note that the running time claimed there is not fully polynomial and provide SWEEP-TREE for completeness.

Using algorithm SWEEP-TREE, 
we next give the algorithm REPAIRMAN-TREE for multiple trimmed windows.
This algorithm uses dynamic programming to move from period to period, in increasing order by 
time.  As it progresses, it finds service runs of all possible profits from every trimmed request in the current 
period through some subset of trimmed requests in the current period and arriving at any possible trimmed request in a later period.  
In this way, for every profit value, we identify the earliest that we can arrive at a request that achieves that profit value.  The 
critical insight is that we may have to leave a certain period rather early in order to reach later requests in time.  By 
recording even such low profit service runs and considering them as starting points, we never rule out a service run that 
appears to be unpromising in early stages but arrives early enough to visit a large number of requests in later stages.
      
We focus on those periods that contain at least one trimmed request and number them 
from $S_1$, the period starting at the smallest time value,
up to the last period $S_m$.  
Let $n$ be the total number of requests.  
For every period $S_i$, we arbitrarily number its trimmed requests as $s_{ij}$.
Let $R_{ij}^k$ be the earliest arriving $k$-profit sequence of service events ending at 
$s_{ij}$.  Let $A_{ij}^k$ be the arrival time of $R_{ij}^k$ at $s_{ij}$.  For each $s_{ij}$, we initialize every $R_{ij}^1$ to be 
$\{s_{ij}\}$ and every $A_{ij}^1$ to be 0.  For $k > 1$, 
let $R_{ij}^k$ be initialized to $null$, 
and let every other $A_{ij}^k$ be initialized to {\it begin}$(S_{i+1})$, 
where {\it begin}$(S_i)$ is the first time instant in period $S_i$.

We use SWEEP-TREE to find a path of shortest length 
from a given starting request to a given ending request,
subject to accumulating a specified profit.
Let $time(R)$ be the amount of time a path $R$ takes.  
For each indexed period $S_i$, from 1 up to $m$, we process period $S_i$ as described in PROCESS-PERIOD.

After all the periods have been processed, we identify the largest-profit path found,
and return that resulting service run $R$ as the output of algorithm REPAIRMAN-TREE.

\begin{theorem}
\label{theorem:repairtree}
In $O(n^4)$ time algorithm REPAIRMAN-TREE finds a service run on a tree that has at least $\frac{1}{3}$ the profit of an optimal service run.
\end{theorem}

\begin{proof}
Correctness follows because the dynamic programming structure of PROCESS-PERIODS finds a run of optimal profit on trimmed time windows.  By the Limited Loss Theorem, trimming time windows reduces the profit found to at worst ${1 \over 3}$ of optimal.  Note that SWEEP-TREE need be run only once per node per period.  Thus, REPAIRMAN-TREE takes $O(n^4)$ time, making $O(n^2)$ calls to SWEEP-TREE,
each of which takes $O(n^2)$ time.\end{proof}

\begin{table}[!hbt]
\begin{tabular}{l}
\smallskip\\
\toprule
\textbf{PROCESS-PERIOD(period $S_i$)} \\
\midrule
\tab For each trimmed request $s_{ij}$ in period $S_i$,\\ 
\tab \tab For each possible profit value $p$,\\
\tab \tab \tab For each subsequent period $S_a$ that contains a trimmed request,\\
\tab \tab \tab \tab For each trimmed request $s_{ab}$ in $S_a$, do the following:\\	
\tab \tab \tab \tab \tab Let $R$ be the path corresponding to $L_r(p)$ that results from \\
\tab \tab \tab \tab \tab \tab SWEEP-TREE($r$) with $s=s_{ij}$ and $t= s_{ab}$, on the set $S_i - \{s_{ij}\}$.\\
\tab \tab \tab \tab \tab Let $R^-$ be $R$ with its last leg, ending at $s_{ab}$, removed.\\
\tab \tab \tab \tab \tab For $k$ from $1$ to $n - \pi(R)$,\\ 
\tab \tab \tab \tab \tab \tab If $A_{ij}^k + time(R^-) < $ {\it begin}$(S_{i+1})$, then\\	
\tab \tab \tab \tab \tab \tab \tab Let profit $q$ be $k + \pi(R) - 1$.\\
\tab \tab \tab \tab \tab \tab \tab If $A_{ij}^k + time(R) < A_{ab}^q$, then\\
\tab \tab \tab \tab \tab \tab \tab \tab Set $R_{ab}^q$ to be $R_{ij}^k$ followed by $R$.\\
\tab \tab \tab \tab \tab \tab \tab \tab Set $A_{ab}^q$ to be $\max \{A_{ij}^k + time(R)$, {\it begin}$(S_a)\}$.\\
\bottomrule
\end{tabular}
\end{table}

\section{Repairman Problem for a Graph}
\label{section:repairman-graph}

In this section we describe our approximation algorithm for the repairman on a graph.  We incorporate improvements from \cite{Chekuri2} into \cite{Bansal}, using profit values rather than distance.  We will also refer to algorithms in \cite{Blum3} that are used as subroutines in \cite{Bansal}.  And, once again, we must embed these techniques into our dynamic programming approach.

Our approximation algorithm for the repairman on a graph
uses approximation algorithms for the following optimization problem.

\begin{description}
\item[Source-Sink $k$-Path ($k$-SSP):]
Given nodes $s$ and $t$ and integer $k$,
find a path of smallest cost from $s$ to $t$ that contains at least $k$ nodes.  (This problem is called min-cost $s$-$t$ path in 
\cite{Blum3}.)
\end{description}  

\noindent 
We consider two approximation problems for $k$-SSP.  Let $c(P)$ be the cost of path $P$ on metric $d$.
Following \cite{Blum3}, let the \emph{excess} of a path $P$ from $s$ to $t$
be $\varepsilon(P) = c(P) - d(s,t)$.  Note that we use $\varepsilon$ to refer to excess while using the visually similar $\epsilon$ to refer to small constants greater than 0.

\begin{description}

\item[Small-Excess $k$-SSP:]
Given nodes $s$ and $t$ and integer $k$,
find a path of small excess from $s$ to $t$ containing at least $k$ nodes.  
(This problem is called min-excess path in \cite{Blum3}.)  

\item[Reduced-Profit $k$-SSP:]
Given nodes $s$ and $t$ and integers $k$ and $\beta > 1$,
find a path from $s$ to $t$ containing at least $k/\beta$ nodes
and costing no more than an optimal $k$-SSP.
\end{description}

The final performance bound of our algorithm for a graph depends on the approximation ratio of the Reduced-Profit $k$-SSP algorithm, which depends primarily on a bicriteria algorithm for Small-Excess $k$-SSP given in \cite{Chekuri2}. 
We describe a technique similar to the one used in \cite{Bansal} to solve the Reduced-Profit 
$k$-SSP problem using this bicriteria approximation to Small-Excess $k$-SSP.  
Finally, we describe how to approximate the repairman problem by
nesting the approximation for Reduced-Profit 
$k$-SSP within our dynamic programming structure.

To follow the approach outlined in the preceding paragraph, we first describe an algorithm we will use as a subroutine.  Let BI-EXCESS be the bicriteria $( {1\over 1-\epsilon}, 2)$-approximation algorithm to Small-Excess $k$-SSP given in \cite{Chekuri2}.  Given a starting node $s$, an ending node $t$, and a profit level $k$, BI-EXCESS returns a path from $s$ to $t$ with profit at least $k(1 - \epsilon)$ whose excess is no more than twice that of an optimal path from $s$ to $t$ collecting profit $k$.

We next describe approximation algorithm REDUCED-PATH
for the Reduced-Profit $k$-SSP problem
which is similar to the one used in \cite{Blum3} for orienteering except that we supply a profit value instead of a distance bound as a parameter and expand the analysis to include bicriteria approximations for Small-Excess $k$-SSP.  To find a path $B$, we 
identify many possible subpaths $B_j$ by choosing all possible pairs of nodes $u$, $v$ and running BI-EXCESS between them with 
profit parameter $k/\beta$.  Of all these possible pairs, we keep the one for which $c(B_j) + d(s, u) + d(v, t)$ is smallest.  
For that pair, we form path $B$ by concatenating edge $(s, u)$, path $B_j$, and edge $(v, t)$.  

\begin{lemma}
\label{lemma:repair profit factor}
Path $B$ from REDUCED-PATH gives a $(2 + \epsilon)$-approximation for Reduced-Profit $k$-SSP.  That is, $\pi(B) \geq {1 \over 2 + \epsilon}\pi(P^*)$ and $c(B) \leq c(P^*)$, where $P^*$ is an optimal solution to $k$-SSP.
\end{lemma}

\begin{proof}
If we wish to find a $(2 + \epsilon)$-approximation for Reduced-Profit $k$-SSP, we compute $\epsilon'$ such that $2 + \epsilon = {2 \over 1 - \epsilon'}$.  Using a generalization of Theorem 1 from \cite{Bansal} with the bicriteria adaptations from \cite{Chekuri2}, a $({1 \over 1 - \epsilon'}, 2)$-approximation for BI-EXCESS allows REDUCED-PATH to produce path $B$ such that $\pi(B) \geq (1 - \epsilon')\pi(P^*_j) \geq {1 -\epsilon'\over 2}\pi(P^*) = {1 \over 2 + \epsilon}\pi(P^*)$.
\end{proof}

  The running time for BI-EXCESS is $\Lambda(n, \epsilon) = O(n^{O(1/\epsilon^2)})$ in \cite{Chekuri2}.  Since our algorithm runs BI-EXCESS for all possible pairs $u$ and $v$ inside a given period, for all possible profit values $k \leq n$, REDUCED-PATH runs in $O(n^3 \Lambda( n, \epsilon ))$.

Finally, we approximate the repairman problem as a whole.  Our approximation algorithm for the repairman on a graph
incorporates the preceding approximation algorithms within the context of
a dynamic programming algorithm with the same overall structure
as the algorithm for a tree.
For each indexed period $S_i$, from 1 up to $m$, we process period $S_i$ as in PROCESS-PERIOD.
The only difference is that instead of
taking $R$ to be the path corresponding to $L_r(p)$ that results from SWEEP-TREE($r$),
it takes $R$ to be the output of REDUCED-PATH($s_{ij}, s_{ab}, p)$ on the set $S_i - \{s_{ij}\}$, 
where $s_{ij}$ is the starting request in the path, $s_{ab}$ 
is the ending request in the path, and $p$ is the profit of which the path must have a constant fraction. 
As before, we identify the largest-profit path found and return the resulting service run $R$ as the output of algorithm REPAIRMAN-GRAPH.

\begin{theorem}
\label{theorem:repairgraph}
In  $O(n^4 \Lambda(n, \epsilon))$ time, REPAIRMAN-GRAPH finds a service run that has at least $\frac{1}{6 + \epsilon}$ the profit of an optimal service run.
\end{theorem}

\begin{proof}
By using the algorithm PROCESS-PERIOD but substituting in the REDUCED-PATH algorithm for the SWEEP-TREE algorithm, REPAIRMAN-GRAPH will find a most profitable service run that is made up of these approximately optimal $k$-SSP paths.  Correctness also follows by the same argument as for Theorem \ref{theorem:repairtree} with the substitution of approximately optimal paths found inside each period.  After taking into account the factor of $2 + \epsilon$ for the $k$-SSP approximation and the factor of $3$ for trimming, the service run $R$ returned by REPAIRMAN-GRAPH has a profit $\pi(R)$ such that $\pi(R) \geq \frac{1}{6 + \epsilon} \pi(R^*)$.

REPAIRMAN-GRAPH will run REDUCED-PATH once for each period to find all the approximately optimal paths.  Since the total number of periods with time windows trimmed into them is no greater than the number of requests, making $O(n)$ calls to REDUCED-PATH will take $O(n^4 \Lambda(n, \epsilon))$ time.
\end{proof}

\section{Deliveryman for a Tree}
\label{section:deliveryman-tree}

With the Small Speedup Theorem at our disposal,
the speeding deliveryman algorithm on a tree
is almost as cleanly conceived as the repairman on a tree.  
When all requests share the same time window, we can find an optimal solution as follows.  For every possible starting request $u$ and ending request $v$ in the period, we identify the direct path between $u$ and $v$.  Remove any leaf and its adjacent edge if the leaf is not $u$ or $v$ or the location of a request, and repeat until every leaf is either $u$ or $v$ or is the location of a request.  We then double up every edge in the slimmed down tree that is not the direct path from $u$ to $v$, and then identify the Euler path from $u$ to $v$.  Since we test all pairs of starting and ending requests, we clearly find the shortest length path and therefore the minimum necessary speed to visit all requests in the tree during a single period.  We can find the direct distances between all pairs in the tree and thus the length of the shortest of the Euler paths in $O(n^2)$ time.

To approximate a solution to the problem on a tree over multiple periods, 
we develop an algorithm to test if a specific speed is fast enough to visit all requests during their periods.  
We use the idea behind the single-period solution in conjunction with dynamic programming.  
TEST-SPEED processes every period in order and finds the earliest-arriving paths starting at request $u$, 
ending at request $v$, and visiting all requests in the period for every pair of requests $u$ and $v$.  
It then glues each of these paths to the earliest-arriving paths which visit all requests in previous periods and, 
for every request $v$ in $S_i$, keeps the earliest-arriving complete path ending at $v$.  
Let $\epsilon$ be an input parameter and $\epsilon' = \epsilon/4$.
Define the algorithm DELIVERY-TREE, which binary searches $\log {1 \over \epsilon'}$ times on a range of speeds using TEST-SPEED, and then, from among the paths found by TEST-SPEED that visit all requests, returns the path that uses the slowest speed.  

For $u \in S_i$, let the arrival time $A_u$ be the earliest time at which a path visiting all the requests in periods before $S_i$ arrives at request $u$ before visiting any other request in $S_i$.  For $v \in S_i$, let the departure time $D_v$ be the earliest time at which a path visiting all requests in the periods up to and including $S_i$ ends at request $v$.  For $u, v \in S_i$, let $length_{uv}$ be the length of the shortest path from $u$ to $v$ which visits all requests in $S_i$.\\

\begin{table}[!hbt]
\begin{tabular}{l}
\toprule
\textbf{TEST-SPEED( \emph{speed} )} \\
\midrule
\tab For each request $u$ in $S_1$, set $A_u$ to be $0$.\\
\tab For $i$ from 1 to $m$ \\
\tab \tab For each request $v$ in $S_i$,\\
\tab \tab \tab Set $D_v$ to be $\min_{u \in S_i} \{A_u + length_{uv} / speed \}$.\\
\tab \tab \tab If $D_v >$ last time instant of $S_i$, then set $D_v$ to be $\infty$.\\
\tab \tab For each request $w$ in $S_{i + 1}$,\\
\tab \tab \tab Set $A_w$ to be $\max \{$ first time instant of $S_{i + 1}, \min_{v \in S_i} \{ D_v + d(v, w) / speed \} \}$\\
\tab \tab \tab If $A_w >$ last time instant of $S_{i + 1}$, then set $A_w$ to be $\infty$.\\
\tab If there exists a request $v \in S_m$ such that $D_v < \infty$, \\
\tab \tab then return ``feasible speed'', else return ``speed too slow''.\\
\bottomrule
\end{tabular}
\end{table}

\begin{theorem}
\label{theorem:deliverytree}
For any $\epsilon > 0$, in $O(n^3 \log {1 \over \epsilon})$ time DELIVERY-TREE finds a service tour of speed at most $4 + \epsilon$ times the optimal speed.
\end{theorem}

\begin{proof}
Correctness follows because an earliest-arriving path found by TEST-SPEED that visits all requests during their periods implies that the speed is sufficiently fast.  
In the next section, we will give a technique which takes $O(n^3)$ time to find a service tour $Q$ 
on trimmed windows such that ${1\over 2}s(Q) \leq s(Q^*) \leq s(Q)$ 
where $Q^*$ is an optimal service tour over trimmed windows.  
Using TEST-SPEED, our algorithm DELIVERY-TREE
binary searches in this range to find a speed within a factor of $1 + \epsilon'$ of the optimal speed.  
We get a total of $O( n^3 \log {1 \over \epsilon })$ time for DELIVERY-TREE.
Using the Small Speedup Theorem applies a factor of 4, yielding a $(4 + \epsilon)$-approximation.
\end{proof}

\section{Deliveryman for a Graph}
\label{section:deliveryman-graph}

The algorithm for the deliveryman on a graph takes direct advantage
of the Small Speedup Theorem, using a minimum spanning tree (MST) algorithm as its main workhorse.
Given an instance of the problem on trimmed windows, we find a near-minimum-speed service tour in the following way.  For each period, we find an MST of the nodes in that period.  For every period $S_i$ but the last, we find the shortest path which connects the MST of the points in $S_i$ to the MST of the points in $S_{i+1}$.  Let the point in $S_i$ which is adjacent to that edge be called $v_i$, and let the point in $S_{i+1}$ be called $u_{i+1}$.

Within each period $S_i$, we double up all edges in the MST that are not 
on the direct path from $u_i$ to $v_i$ and sequence all edges into an Euler path.  
We then connect these $m$ Euler paths by the edges $(v_i, u_{i+1})$ for $1 \leq i < m$.  
Thus, we have created a service tour from $u_1$ to $v_m$.  
We then determine the minimum speed at which this service tour can be taken 
and still visit all requests during their trimmed time windows.

Let $c(u_i, v_j)$ denote the cost of traveling from $u_i$ to $v_j$ along the service tour.  For all pairs $i$ and $j$ where $1 \leq i \leq m$ and $i \leq j \leq m$, we find the speed needed to cover the tour from $u_i$ to $v_j$ and store these speeds in a set.  The value of each speed will be ${2 \cdot c(u_i, v_j) \over j - i + 1}$ as the factor of 2 accounts for the .5 unit-time windows.  We search in this list until we find the lowest speed at which we can visit all requests within their periods.  As one of these pairs of nodes must define the most constraining speed, we find the minimum speed at which we can travel $Q$.

The total running time for our algorithm DELIVERY-GRAPH is $O( n^3 )$.  For each node in each period, DELIVERY-GRAPH will run a single-source shortest path algorithm to find the weights needed to construct the MSTs, totaling $O( n^3 )$.  Then, all trees can be built and connected in $O( n^2 )$.  Finding the speeds for all pairs can also be done in $O(n ^2 )$.

Let $Q^*$ be the optimal tour over trimmed windows, and let $Q$ be the tour generated by our algorithm.  
Let $Q^*_i$ be the subtour of tour $Q^*$ restricted to requests inside period $S_i$,
and $Q_i$ be the subtour of tour $Q$ restricted to $S_i$.   Let $u^*_i$ be the first node in subtour $Q^*_i$ and $v^*_i$ the last.  Similarly, let $u_i$ be the first node in subtour $Q_i$ and $v_i$ the last.

\begin{lemma}
\label{lemma:deliveryspeed}
Let $s$ be the minimum required speed for $Q^*$.  Suppose a deliveryman $M^*$ is traveling along $Q^*$ at speed $\eta$ while another deliveryman $M$ is traveling along $Q$ at a speed that never exceeds $2\eta$.  Then for each $S_i$ deliveryman $M^*$ will never arrive at the first node in $Q^*_i$ before deliveryman $M$ arrives at the first node in $Q_i$.
\end{lemma}
\begin{proof} By induction on $i$.

\noindent Basis: ($i = 1$)

Deliverymen $M$ and $M^*$ start tours $Q$ and $Q^*$, respectively, at the same time.

\noindent Induction Step: ($i > 1$)

The length of subtour $Q^*_{i-1}$ is at least half the length of subtour $Q_{i-1}$.  The distance from $v^*_{i-1}$ to $u^*_i$ is never shorter than the distance from $v_{i-1}$ to $u_i$.  Thus, the total distance from $u^*_{i-1}$ to $u^*_i$ in $Q^*$ is never less than half the total distance from $u_{i-1}$ to $u_{i}$ in $Q$.  Since deliveryman $M$ is traveling $Q$ at a maximum speed which is twice the maximum speed that deliveryman $M^*$ travels $Q^*$, $M$ arrives at $u_i$ no later than $D^*$ arrives at $u^*_i$.
\end{proof}

\begin{theorem}
\label{theorem:deliverygraph}
In $O(n^3)$ time DELIVERY-GRAPH finds a service tour of speed at most 8 times the optimal speed.
\end{theorem}

\begin{proof}
Correctness follows because DELIVERY-GRAPH finds the lowest possible speed by examining all pairs of nodes.  The factor of 2 given by Lemma \ref{lemma:deliveryspeed} multiplied by the factor of 4 given by the Small Speedup Theorem for trimming shows that DELIVERY-GRAPH returns a service tour $Q$ such that $s(Q) \leq 8 \ s(Q^*)$, 
where $Q^*$ is an optimal service tour.
\end{proof}

\section{NP-hardness on a Tree}
\label{section:np-hardness}

By a reduction to TSP, the traveling repairman problem with unit-time windows is APX-hard on a weighted, metric graph.
For the case of a line, NP-completeness proofs for many of the time-constrained traveling salesman problems were given in \cite{Tsitsiklis}, 
but we know of no proof that the unit-time window repairman problem on a line given in \cite{Bar-Yehuda} is NP-hard.  Below, we consider the hardness of repairman when the problem is on a tree.

\begin{theorem}
The traveling repairman problem with unit-length time windows whose nodes are connected by a tree-shaped network is NP-hard.
\end{theorem}

\begin{proof}
We use a reduction from a version of the partition problem restricted to positive values:  
Given a multiset of $n$ positive integers, 
decide whether the multiset can be partitioned into two multisets which sum to the same value, 
i.e.~half the sum of all of the integers in the multiset.  This problem is NP-complete \cite{Karp}.

Our reduction is as follows.  
We assume that the sum of the integers in the multiset is $2K$.  
First create a central node $u$ in a tree by itself.  
For each integer in the multiset, create a node and connect it to $u$ with an edge having the cost of the integer.  
Then create the start and end nodes $s$ and $t$ and connect them to $u$ with edges of cost $6K$.  
Also create the midpoint node $v$ and connect it to $u$ with an edge of cost $K$.  
To complete the input for the repairman problem, we must also create service requests having both a time window and a node for a location.  
Create a service request with time window $[0, 6K]$ located at $s$ and a request with window $[12K, 18K]$ located at $t$.  
Create requests located at each of the nodes corresponding to an integer and also at the central node $u$ all with time window $[6K, 12K]$.  
Finally, create two requests located at the midpoint node $v$, one with time window $[3K, 9K]$ and the other with time window $[9K, 15K]$.  Recall that our definition of service requests allows multiple requests to share a single node as a location.  
Note that the graph is a tree and that all the time windows are exactly $6K$ units long, meeting the unit-length requirement.  
The optimal tour has just enough time to visit all of the nodes if and only if 
it starts at the start node, visits a set of nodes whose integers sum to $K$, visits the midpoint node, 
visits the remaining nodes which also sum to exactly $K$, and finally ends at the end node.
\end{proof}

While this proof assumes intervals closed on both ends,
it is trivial to modify the proof for intervals closed on one end and open on the other.
A proof that the speeding deliveryman problem on a tree with unit-time windows is NP-hard follows the same form.

\section{Handling Nonzero Service Times}
\label{section:service-times}

Some versions of the traveling repairman problem may require a non-zero service time for each service event.  We briefly discuss two natural models for service times.  In the first, the interval of service time must be completely contained within the time window for the service event.  In the second, the interval of service time needs only to start within the corresponding time window and but not necessarily to finish within the time window.

In the first model, we can easily add a uniform service time $\mu < 1$ to our solution.  For each request, we create a new node connected only to the request node under consideration with an edge requiring ${ \mu \over 2}$ time to cross at the given speed.  Given that the original request had a time window of $[t, t + 1)$, we add a request at the new node with a time window of $[t + {\mu \over 2}, t + 1 - {\mu \over 2})$.  Then, we remove the original request.  By this construction, we guarantee that there is a delay of $\mu$ after visiting the original request and that the original request is visited in the correct interval of $[t, t + 1 - \mu)$.  Also, after preprocessing, all time windows will be of length $1 - \mu$, satisfying the uniform time window requirement for our algorithms.

In the second model, we can add arbitrary length service times in a similar way.  Consider a single service request $r$ to which we wish to add a service time of length $\mu_r$.  We again add a node connected only to the node in question with an edge requiring $\mu_r/2$ time to cross at the given speed.  Given that the original request has a time window of $[t, t + 1)$, we add a request to the new node with a time window of $[t + \mu_r/2, t + 1 + \mu_r/2)$ and then remove the original request.  Again, this construction introduces the necessary delays while maintaining the original constraints and uniform time windows.

\section{Repairman with Windows of Different Length}
\label{section:different-repairman}

In this section we present an algorithm that achieves a constant-factor approximation for the traveling repairman on windows with length between 1 and 2 and then explain how the same ideas can be extended to general time windows. For integral edge lengths and window release and deadline times, $O(\log^2 n)$ and $O(\log D_{\max})$-approximations are given in \cite{Bansal} for the general, rooted repairman problem, where $D_{\max}$ is the latest time a time window ends.  After the initial publication of our work on unit time windows in \cite{Frederickson3}, an extension to windows with length between 1 and 2 was given in \cite{Chekuri3} for the unrooted repairman problem.  Under the same assumptions, the better bound of $O(\log L)$ is achieved in \cite{Chekuri3}, where $L$ is the longest time window.

Here we present improved approximation factors for the case that windows have length between 1 and 2.  It was also claimed in \cite{Chekuri3} that a constant approximation for this case allows a $O(\log D)$-approximation for the general, unrooted repairman problem, where $D$ is the ratio of the length of the longest time window to the length of the shortest time window.  We give an algorithm and analysis for the general, unrooted problem with improved constants and an $O(\log_b D)$-approximation for any fixed base $b$ of the logarithm.

Below we describe the algorithm WINDOW12 which allows us to approximate the repairman problem when windows have length between 1 and 2.  To unify notation, let $\Gamma(n)$ represent the running time for repairman approximations with trimmed windows on either a metric graph or a tree, as appropriate.  As shown in Theorems \ref{theorem:repairtree} and \ref{theorem:repairgraph}, $\Gamma(n)$ is $O(n^4)$ for a tree and $O(n^4 \Lambda(n, \epsilon))$ for a metric graph.  Likewise, let the approximation ratios for the repairman algorithms used after trimming be $\gamma$, where $\gamma = 1$ for a tree, where $k$-SSP can be solved optimally, and $\gamma = {2 + \epsilon}$ for a metric graph, as shown in Lemma \ref{lemma:repair profit factor}.

As pointed out in \cite{Chekuri3}, a relatively simple extension of our Limited Loss Theorem allows one to achieve a ${5\gamma}$-approximation when windows have length between 1 and 2.  A more refined approach that we now describe will improve substantially on this constant.  Our approach is to try several different sizes for periods.  When most of the windows are of length closer to 1, then a period size of 1/2 works well.  When most of the windows are of length closer to 2, then a period size of 1 works well.  When many of the windows are of length closer to 3/2, then a period size of 3/4 works well.

For each period size, we will consider multiple starting points for a set of periods, each spaced 1/4 apart.  Thus, sets of periods whose period sizes are 1/2, 3/4, and 1 will have 2, 3, and 4 unique starting positions, respectively.  Depending on a given period size and starting point, a window will partially fill 2 subintervals and fully fill 0, 1, 2, or 3 subintervals between the 2 partial intervals.  Let $W_\ell$ be the set of windows
that completely fills exactly $\ell$ subintervals and partially overlaps with two more of them.

When trimming, we may have to select from among several choices of which single full subinterval to keep for each window.  For example, for periods of length 1/2 and for windows in $W_3$ which would have three full subintervals, the choices for trimming will be the first, second, or third full subinterval.  Combining these choices with the two choices associated with windows in $W_2$ and the single choice in windows in $W_1$ would yield 6 trimmings, or, in general, $k!$ where $k$ is the largest number of subintervals completely filled.  Let REPAIR be the appropriate basic repairman algorithm on trimmed windows, either for a tree or for a metric graph, described in Sect.~\ref{section:repairman-tree} or \ref{section:repairman-graph}.  For each period size, for each starting point, for each choice of trimming, we will run REPAIR and keep the result if the profit is better than a previous run.

\begin{table}[!htb]
\begin{tabular}{l}
\toprule
\textbf{WINDOW12} \\
\midrule
\tab \emph{PHASE 1:}\\
\tab Set the period size to 1/2 and identify windows for sets $W_1$, $W_2$, and $W_3$.\\
\tab For $i$ from $0$ to $1$,\\
\tab \tab Set the starting point for the periods to $i/4$.\\
\tab \tab For $j$ from $1$ to 2,\\
\tab \tab \tab For $k$ from $1$ to 3,\\
\tab \tab \tab \tab Trim each window in $W_1$ to its $1^{st}$ full subinterval.\\
\tab \tab \tab \tab Trim each window in $W_2$ to its $j^{th}$ full subinterval.\\
\tab \tab \tab \tab Trim each window in $W_3$ to its $k^{th}$ full subinterval.\\
\tab \tab \tab \tab Run REPAIR and retain the best result so far.\\
\tab \emph{PHASE 2:}\\
\tab Set the period size to 3/4 and identify windows for $W_1$ and $W_2$.\\
\tab For $i$ from $0$ to $2$,\\
\tab \tab Set the starting point for the periods to $i/4$.\\
\tab \tab For $j$ from $1$ to 2,\\
\tab \tab \tab Trim each window in $W_1$ to its $1^{st}$ full subinterval.\\
\tab \tab \tab Trim each window in $W_2$ to its $j^{th}$ full subinterval.\\
\tab \tab \tab Run REPAIR and retain the best result so far.\\
\tab \emph{PHASE 3:}\\
\tab Set the period size to 1 and identify windows for $W_1$.\\
\tab For $i$ from $0$ to $3$, \\
\tab \tab Set the starting point for the periods to $i/4$.\\
\tab \tab Trim each window in $W_1$ to its $1^{st}$ full subinterval.\\
\tab \tab Run REPAIR and retain the best result so far.\\
\bottomrule
\end{tabular}
\end{table}

When window lengths are not all the same, our analysis depends on an averaging argument.  By using many service runs based on an optimal run, we can record the total number of times a given interval is visited by all runs.  From all the intervals of all the windows, we find one that is visited the least.  The number of times this interval is visited divided by the total number of runs is a lower bound on the fraction of profit collected, relative to optimal.

We analyze the performance of WINDOW12 as follows.
Let $R^*$ be an optimal service run for a repairman instance with time window lengths from 1 up to but not including 2.
For the sake of analysis, we introduce a new set of periods with duration $1/4$. If we split each window into subintervals of length $1/4$ along boundaries of these new periods, we get windows in sets $H_3$, $H_4$, $H_5$, $H_6$, and $H_7$.  Let the total fraction of profit in an optimal solution coming from windows in set $H_\ell$ be $h_\ell$.  Thus, $\sum_{\ell = 3}^7 h_\ell = 1$.

We use these subintervals to give a finer granularity when analyzing the performance of the algorithm run on periods of greater length, viz. $1/2$, $3/4$, and $1$.  Consider set $H_\ell$ of windows, $\ell = 3, 4, 5, 6, 7$ and period length $j/4$, for $j = 1, 2, 3, 4$.

Then, the number of full subintervals of a window in $H_\ell$ that are covered when the period length is $j/4$ is either
$\left\lfloor{(\ell - j - 1)/j}\right\rfloor \text{ or } \left\lceil{(\ell - j - 1)/j}\right\rceil$
depending on which set of periods is used.  The average coverage over all sets of periods with period length $j/4$ is 
${(\ell - j - 1)/j}$.  As before, let $\gamma$ be the approximation bound on the basic repairman algorithm on unit-time windows when trimming has already been done.

\begin{lemma}
\label{lemma:inverse}
Let $w_\ell$ be the fraction of total profit gained by visiting windows of type $W_\ell$ with an optimal path on untrimmed windows.  Let $\mathcal{W}$ be a collection of sets $W_\ell$ given by trimming windows in an instance of the repairman problem.  The fraction of profit for the best run on trimmed windows in that instance is at least 
$${1 \over \gamma} \sum_{W_\ell \in \mathcal{W}} {w_\ell \over \ell + 2}$$
\end{lemma}

\begin{proof}
Each window in set $W_\ell$ is divided into $\ell + 2$ different subintervals.  Label the fraction of profit from each of these respective subintervals $w_\ell^{(1)}$ through $w_\ell^{(\ell + 2)}$.  Note that $\sum_{i = 1}^{\ell + 2} w_\ell^{(i)} = w_\ell$.

Let $W_k$ be the set of windows in $\mathcal{W}$ which can be divided into the largest number of subintervals.  For $W_k$, there is a subinterval $i$ such that $w_k^{(i)} \leq w_k/(k + 2)$.  Ignore that subinterval and pair up the $k + 1$ subintervals from $W_{k - 1}$ in increasing order of index with the $k + 1$ remaining subintervals from $W_k$.  The total profit from these $k + 1$ pairs is at least $\sigma  = w_{k - 1} + (k + 1)w_k/(k + 2)$.  Of these pairs, there is a pair whose profit is no greater than $\sigma/(k + 1)$.  Then, ignore this pair and, in increasing order of index, match up the $k$ subintervals from $W_{k - 2}$ with the $k$ remaining pairs.  Repeat this process of ignoring a smallest tuple for a given $\ell$ and matching it up with the $\ell + 1$ subintervals from the next set of smaller windows $W_{\ell - 1}$.

At the end of the process, the remaining $k$-tuple will have a fraction of profit that is

$$\sum_{\ell = 1}^k {(\ell + 1)!w_\ell \over (\ell + 2)!} = \sum_{W_\ell \in \mathcal{W}} {w_\ell \over \ell + 2}$$
\end{proof}

\begin{lemma}
\label{lemma:step 3}
From among the two sets of periods and among the six different trimmings created in the first phase of WINDOW12, one such pair of choices yields a run $R$ such that $\pi(R)/\pi(R^*) \geq (h_3/3 + 7h_4/24 + h_5/4 + 9h_6/40 + h_7/5)/\gamma$.
\end{lemma}

\begin{proof}
Consider the trimmed windows from two different shifts of periods from the first phase of WINDOW12.  By application of Lemma \ref{lemma:inverse}, we find the following contributions.  Windows from $H_3$ will contribute $h_3/3$ in both sets of periods.  Windows from $H_4$ will contribute $h_4/3$ in one set of periods and $h_4/4$ in the other.  Windows from $H_5$ will contribute $h_5/4$ in both sets of periods.  Windows from $H_6$ windows will contribute $h_6/4$ in one set of periods and $h_6/5$ in the other.  Finally, windows from set $H_7$ windows will contribute $h_7/5$ in both sets of periods.  When the values for both sets of periods are averaged together, the final result satisfies $\pi(R)/\pi(R^*) \geq (h_3/3 + h_4/6 + h_4/8 + h_5/4 + h_6/8 + h_6/10 + h_7/5)/\gamma = (h_3/3 + 7h_4/24 + h_5/4 + 9h_6/40 + h_7/5)/\gamma$.
\end{proof}

\begin{lemma}
\label{lemma:step 8}
From among the three sets of periods and among the two different trimmings created in the second phase of WINDOW12, one such pair of choices yields a run $R$ such that $\pi(R)/\pi(R^*) \geq (h_3/9 + 2h_4/9 + h_5/3 + 11h_6/36 + 5h_7/18)/\gamma$.
\end{lemma}

\begin{proof}
Consider the trimmed windows from three different shifts of periods from the second phase of WINDOW12.  By application of Lemma \ref{lemma:inverse}, we find the following contributions.  Windows from $H_3$ will contribute $h_3/3$ in one set of periods and nothing in the other two.  Windows from $H_4$ will contribute $h_4/3$ in two sets of periods and nothing in the other one.  Window from $H_5$ will contribute $h_5/3$ in all three sets of periods.  Windows from $H_6$ will contribute $h_6/3$ in two sets of periods and $h_6/4$ in the other one.  Windows from $H_7$ will contribute $h_7/3$ in one set of periods and $h_7/4$ in the other two.  Averaging the values for all three sets of periods gives the claimed result.
\end{proof}

\begin{lemma}
\label{lemma:step 13}
From among the four sets of periods created in the third phase of WINDOW12, one of them yields a run $R$ such that $\pi(R)/\pi(R^*) \geq (h_4/12 + h_5/6 + h_6/4 + h_7/3)/\gamma$.
\end{lemma}

\begin{proof}
Consider the trimmed windows from four different shifts of periods from the third phase of WINDOW12.  By application of Lemma \ref{lemma:inverse}, we find the following contributions.  Windows from $H_3$ will contribute nothing in all four sets of periods.  Windows from  $H_4$ will contribute $h_4/3$ in one set of periods and nothing in the other three.   Windows from $H_5$ will contribute $h_5/3$ in two sets of periods and nothing in the other two.  Windows from $H_6$ will contribute $h_6/3$ in three sets of periods and nothing in the other one.  Windows from $H_7$ will contribute $h_7/3$ in all four sets of periods.  Averaging the values for all four sets of periods gives the claimed result.
\end{proof}

\begin{theorem}
In $O(\Gamma(n))$ time, algorithm WINDOW12 identifies a run $R$ such that\\ $\pi(R)/\pi(R^*) \geq 52/(219\gamma)$.
\end{theorem}

\begin{proof}
By Lemmas \ref{lemma:step 3}, \ref{lemma:step 8}, and \ref{lemma:step 13}, $\pi(R)/\pi(R^*) \geq {1 \over \gamma} \max\{h_3/3 + 7h_4/24 + h_5/4 + 9h_6/40 + h_7/5$, $h_3/9 + 2h_4/9 + h_5/3 + 11h_6/36 + 5h_7/18$, $h_4/12 + h_5/6 + h_6/4 + h_7/3\}$.  Thus, for any convex combination using positive $x$, $y$, and $z$, where $x + y + z = 1$, we have 
\begin{eqnarray*}
\pi(R)/\pi(R^*) \geq {1 \over \gamma} \max_{x + y + z = 1}\Big\{&x(h_3/3 + 7h_4/24 + h_5/4 + 9h_6/40 + h_7/5)&\\
+ &y(h_3/9 + 2h_4/9 + h_5/3 + 11h_6/36 + 5h_7/18)&\\
+ &z(h_4/12 + h_5/6 + h_6/4 + h_7/3)&\Big\}
\end{eqnarray*}
This expression achieves a maximum for $x = 50/73$, $y = 6/73$, and $z = 17/73$, yielding

$$\pi(R)/\pi(R^*) \geq {52(h_3 + h_4 + h_5 + h_6 + h_7)\over 219\gamma } = {52 \over 219\gamma  }$$

Algorithm WINDOW12 runs the basic repairman algorithm 12 times in the first phase, 6 times in the second phase, and 4 times in the third phase, for a total of 22 times.  Since the running time of the basic repairman algorithm is $\Gamma(n)$, the running time of WINDOW12 is $O(\Gamma(n))$.
\end{proof}

For handling requests whose windows are between 1 and 2, the best performance ratio that we have achieved is $219\gamma/52$.  To handle windows whose largest length $D$ is either greater than or less than 2, we have identified two additional techniques.

The first technique, applicable for $D > 2$, partitions the windows into sets such that the lengths within each set will be within a factor of 2 of each other.  We then run WINDOW12 on each such set and choose the best result.  This idea of partitioning windows into sets corresponding to lengths in ranges bounded by consecutive powers of 2 was also used in \cite{Bar-Yehuda}. 

The second technique, applicable to either $D < 2$ or $D > 2$, extends the approach of WINDOW12.  For $D$ sufficiently smaller than 2, we will use a proportionately smaller distance between the beginnings of periods, and, rather than use periods of length 1/2, 3/4, and 1, use periods of length 1/2, $1/2 + 1/2^k$, and $1/2 + 2/2^k$.  For $D$ sufficiently larger than 2, we will use more than three different period lengths.  For $D = 3$, for example, we would use period lengths of 1/2, 3/4, 1, 5/4, and 3/2.

As the number of period lengths increases, the time to consider various combinations grows exponentially.  It thus makes sense, when $D > 2$, to use a combination of the first and second techniques.  We will analyze the performance with this in mind.

We now present a generalized version of WINDOW12 called WINDOWG which can be applied to windows of any length between 1 and $1 + p/2^g$ for natural numbers $p$ and $g$.  Let $q = 1/2^{g + 1}$.
Let $P_0$ be a set of periods of length $q$.

For $i = 0$, $1$, $\ldots$, $p$, let $P_i^{(0)}$, $P_i^{(1)}$, $\ldots$, $P_i^{(i)}$ be sets of periods of length $(i + 2^g)q$, where the first period of $P_i^{(0)}$ begins at the same instant as the first period of $P_0$, and for each $j = 1$, 2, $\ldots$, $i$, the first period of $P_i^{(j)}$ begins $q$ after the first period of $P_i^{(j - 1)}$.

It is convenient to use the \emph{factorial number system} \cite[p.175]{Knuth}, which we review here.  In this system, a nonnegative integer is represented by a sequence of digits $d_u \ldots d_2 d_1$ where $d_i \in \{ 0$, $1$, $\ldots$, $i \}$.  The value of $v$ is $\sum_{i = 1}^u i!~d_i$.  Every value is uniquely represented, and it follows that $1\underbrace{0\ldots0}_{u - 1}$ is $u!$.

\begin{lemma}
For integers $g \geq 1$ and $p \geq 1$,  the running time of WINDOWG is $O( (p + g)!~\Gamma(n) )$.
\end{lemma}

\begin{proof}

\noindent From inspection, the running time of WINDOWG is proportional to at most $$\sum_{i = 0}^p (i + 2) \sum_{k = 0}^{(p + g - i)! - 1} \Gamma(n)~~=~~\Gamma(n)( p + g)! \sum_{i = 0}^p (i + 2){(p + g - i)! \over (p + g)!}$$

\noindent We bound the value of the summation by a constant.  Since $p + g - (i - 1) \geq g + 1 \geq 2$ for any $i \leq p$, 

$$\sum_{i = 0}^p (i + 2){(p + g - i)! \over (p + g)!} \leq \sum_{i = 0}^p {(i + 2) \over 2^i} < \sum_{i = 0}^\infty {(i + 2) \over 2^i} = 6$$
\end{proof}

\begin{table}[!htb]
\begin{tabular}{l}
\smallskip\\
\toprule
\textbf{WINDOWG} \\
\midrule
\tab For $i$ from $0$ to $p$ (consider periods of length $(i + 2^g)q$),\\
\tab \tab For $j$ from $0$ to $i + 1$ (consider sets $P_i^{(j)}$ of periods),\\
\tab \tab \tab For $k$ from $0$ to $(p + g - i)! - 1$ (choose trim positions), \\
\tab \tab \tab \tab Let $d_u\ldots d_1$ be the representation of $k$ in the factorial number system.\\
\tab \tab \tab \tab  Assume $d_0 = 0$.\\
\tab \tab \tab \tab  For each request $x$,\\
\tab \tab \tab \tab  \tab Let $v$ be the number of periods of length $(i + 2^g)q$ that $x$ has in $P_i^{(j)}$.\\
\tab \tab \tab \tab  \tab If $v > 0$ then\\
\tab \tab \tab \tab  \tab \tab Let $w = 1 + d_{v - 1}$.\\
\tab \tab \tab \tab  \tab \tab Trim the window of $x$ to its $w^{th}$ period of length $(i + 2^g)q$ in  $P_i^{(j)}$.\\
\tab \tab \tab \tab  \tab Else\\
\tab \tab \tab \tab  \tab \tab Exclude $x$ from this run.\\
\tab \tab \tab \tab  Run REPAIR and retain the best result so far.\\
\bottomrule
\end{tabular}
\end{table}

Note that when $p \geq 2$, there appears to be no benefit to having a value of $g > 1$ except where needed to specify a precise fractional value for maximum window length.  We do not suggest allowing $p$ or $g$ to range freely, since factorial growth is unacceptable.  However, we can fix values of $p$ and $g$ to give an appropriate running time and then partition the windows into sets such that the first partition contains windows whose lengths are between 1 and $1 + p/2^g$, the second set contains windows whose lengths are between $1 + p/2^g$ and $(1 + p/2^g)^2$, the third between $(1 + p/2^g)^2$ and $(1 + p/2^g)^3$, and so on.  Let $D$ be the ratio of the longest window length to the shortest.  Let $b$ be $1 + p/2^g$.  Then, there are $\log_b D$ such sets.  Let WINDOWGD be the algorithm that runs WINDOWG on each of the $\log_b D$ sets separately and returns the highest profit run found. 

\begin{theorem}
Let $p$ and $g$ be fixed and $b = 1 + p/2^g$.  Let $D$ be the ratio of the longest window length to the shortest.  Let the approximation ratio of WINDOWG for windows between $1$ and $b$ be a function of $p$ and $g$ given by $\rho(p,g)\gamma$.  The approximation ratio of WINDOWGD for windows of general length is $\rho(p,g)\gamma/\log_b D$.
\end{theorem}

\begin{proof}
Since the union of all of the $\log_b D$ disjoint sets is the set of all windows, one set must contain at least ${1/\log_b D}$ of the requests serviced by an optimal run.  Since the approximation ratio for each set is $\rho(p,g)\gamma$, the profit found on that set is at least $\rho(p,g)\gamma/\log_b D$.
\end{proof}

Values of $\rho(p,g)$ can be calculated for inputs $p$ and $g$ using a linear program.  For example, for maximum window sizes of 2, 3, and 4, the values of $\rho(p,q)$ are ${52\over 219} \approx .2374$, ${4954\over 24619} \approx .2012$, and ${258044\over 1427019} \approx .1808$, respectively.

\section{Deliveryman with Windows of Different Lengths}
\label{section:different-deliveryman}

In this section we expand our analysis for the deliveryman problem to windows with length between 1 and 2, and also to windows with length in any bounded range.  To unify notation, let $\Delta(n)$ represent the running time for deliveryman approximations using trimmed windows either on a metric graph or on a tree, as appropriate.  As shown in Theorems \ref{theorem:deliverytree} and \ref{theorem:deliverygraph}, $\Delta(n)$ is $O(n^3 \log {1\over \epsilon})$ for a tree and $O(n^3)$ for a metric graph.  Likewise, let the approximation ratios for the algorithms used after trimming be $\delta$, where $\delta = 1 + \epsilon$ for a tree and $\delta = 2$ for a metric graph, as shown by the same theorems.

We now consider the version of the deliveryman problem which allows the lengths of request time windows to range over the interval $[1, 2)$ instead of being confined to unit size.  We trim in exactly the same way we did for the repairman problem with time windows in this range.  Of course, trimming may increase the necessary speed of the best service tour, but by no more than a constant factor.

\begin{lemma}\label{lemma:newtrimming-deliveryman} Let $Q^*$ be an optimal service tour with respect to untrimmed requests whose lengths are in the interval $[1, 2)$.  There exists a service tour $Q$ with respect to trimmed requests such that $s(Q) \leq 6 s(Q^*)$.  
\end{lemma}

\begin{proof}
The proof is similar to that of Theorem \ref{theorem:deliverytrimming}, with minor changes to ensure that $Q^*$ hits all five, instead of what was previously three, intervals that correspond to the at most five .5 length periods with which a time window could intersect.  

We shall extend $Q^*$ backward for $t < 0$ by assuming that $Q^*$ proceeds from any convenient position so that it encounters the original starting position at time $t = 0$.  
Let racing now be movement, either forward or backward, along $Q^*$ at a speed of $6 s(Q^*)$ instead of $4 s(Q^*)$.  We define tour $Q$ which races along $Q^*$.  
During any two consecutive periods, the deliveryman 
will make a net advance equal to the advance of $Q^*$ over those two periods. 

Identify as $t_i$ the time $t = .5i$ which is also the starting time of period $i$. 
We define $Q$ as follows.  Start tour $Q$ at $t = 0$ at the location that $Q^*$ has at time $t = -.5$.  From there, tour $Q$ follows a repeating pattern of racing forward along $Q^*$ for 1 period, racing backward along $Q^*$ for ${5\over 6}$ periods, and racing forward again along $Q^*$ for a final ${1\over 6}$ periods. 
\end{proof}

Let DELIVERY be the appropriate deliveryman algorithm on trimmed windows, either for a tree or for a metric graph, described in Sect.~\ref{section:deliveryman-tree} or \ref{section:deliveryman-graph}.

\begin{theorem}
\label{theorem:delivery12}
In $O(\Delta(n))$ time DELIVERY finds a service tour of speed at most $6\delta$ times the optimal speed for windows with length between 1 and 2.
\end{theorem}

\begin{proof}
Correctness follows from Theorems \ref{theorem:deliverytree} and \ref{theorem:deliverygraph}.  The factor of 6 for trimming given by Lemma \ref{lemma:newtrimming-deliveryman} multiplied by the appropriate approximation factor $\delta$ after trimming shows that DELIVERY returns a service tour $Q$ such that $s(Q) \leq 6\delta\hspace{1pt}s(Q^*)$, 
where $Q^*$ is an optimal service tour.  Since a deliveryman algorithm on trimmed windows is only run a single time, the running time is $O(\Delta(n))$.
\end{proof}

Observe that this pattern can be extrapolated to windows with arbitrary lengths.  As before, let $D$ be the ratio of the longest time window to the shortest.  In the event that the ratio is not an integer, let $D$ be the ceiling of the ratio.  Let the speed be $2D + 2$.  Define $Q$ as a generalization of the previous definition.  Start tour $Q$ at $t = 0$ at the location that $Q^*$ has at time $t = -.5$.  From there, tour $Q$ follows a repeating pattern of racing forward along $Q^*$ for 1 period, racing backward along $Q^*$ for ${2D + 1\over 2D + 2}$ periods, and racing forward again along $Q^*$ for a final ${1\over 2D + 2}$ periods.  Clearly, $Q$ will hit all $2D + 1$ intervals of length .5 that a window can intersect with while making the required progress of $(2D + 2)\Big(1 - {2D + 1\over 2D + 2} + {1\over 2D + 2}\Big) = 2$ periods.

\bibliography{bibliography}

\end{document}